\documentclass[a4paper,UKenglish,cleveref,nameinlink,colorlinks=true,linkcolor=blue,citecolor=blue,urlcolor=black,autoref]{lipics-v2021}
\usepackage{graphicx}
\usepackage[inline]{enumitem}

\newcommand{\leaves}{\ensuremath{L}}
\newcommand{\oh}{\ensuremath{\mathcal{O}}}
\newcommand{\p}{\ensuremath{\mathcal{P}}}
\newcommand{\q}{\ensuremath{\mathcal{Q}}}

\crefname{enumi}{Invariant}{Invariants}

\graphicspath{{figures/}}

\usepackage{booktabs}

\PassOptionsToPackage{colorlinks=true,urlcolor=black,linkcolor=black,citecolor=black}{hyperref}

\newcommand{\prooflink}[1]{\hypersetup{linkcolor=lipicsBulletGray}\hyperref[#1]{$\blacktriangledown$}}
\newcommand{\statlink}[1]{\hypersetup{linkcolor=lipicsYellow}\hyperref[#1]{$\blacktriangle$}}

\newcommand{\cprooflink}[1]{\hypersetup{linkcolor=lipicsBulletGray}\hyperref[#1]{$\triangledown$}}
\newcommand{\cstatlink}[1]{\hypersetup{linkcolor=lipicsYellow}\hyperref[#1]{$\vartriangle$}}

\title{Embedding Trees and Cactus Graphs with Integer Edge Lengths on a Compact Integer Grid}

\author{Aaron B\"ungener}{Technical University of Munich, Heilbronn, Germany}{aaron.buengener@tum.de}{}{}
\author{Henry F\"orster}{John Cabot University, Rome, Italy}{henry.foerster@johncabot.edu}{0000-0002-1441-4189}{}
\author{Stephen Kobourov}{Technical University of Munich, Heilbronn, Germany}{stephen.kobourov@tum.de}{0000-0002-0477-2724}{}
\author{Jacob Miller}{Technical University of Munich, Heilbronn, Germany}{jacob.miller@tum.de}{0000-0002-0567-785X}{}
\author{Johannes Zink}{Technical University of Munich, Heilbronn, Germany}{johannes.zink@tum.de}{0000-0002-7398-718X}{}

\authorrunning{Aaron Büngener, Henry Förster, Stephen Kobourov, Jacob Miller, Johannes Zink}

\Copyright{Aaron Büngener, Henry Förster, Stephen Kobourov, Jacob Miller, Johannes Zink}


\hideLIPIcs
\nolinenumbers


\begin{document}

\maketitle

\begin{abstract}
    Every planar graph can be embedded with crossing-free, straight-line edges on an $\oh(n) \times \oh(n)$ integer grid, where $n$ is the number of vertices.
    It is not known, however, whether integer edge lengths can be guaranteed for all planar graphs (conjectured to be true by Harborth).
    Kleber conjectures something even stronger: \emph{Every planar graph admits a planar straight-line grid drawing where every edge has integer length.}
    Positive results for Kleber's conjecture are known for several classes of planar graphs,
    however, all but one of the existing results are existential and do not provide bounds on the size of the integer grid.
    We address this gap for trees and cactus graphs.
    For trees, we present a linear-time algorithm that computes
    a crossing-free straight-line drawing with integer edge lengths
    on an $\oh(n \log n) \times \oh(n \log n)$ grid.
    We complement this result by giving an $\Omega(n^2)$ lower bound.
    For cactus graphs, we present a linear-time algorithm that computes
    a crossing-free straight-line drawing with integer edge lengths
    on an $\oh(n^2 \log n) \times \oh(n^2 \log n)$ grid.
\end{abstract}

\section{Introduction}

Wagner in 1936~\cite{wagner1936bemerkungen}, F\'ary in 1948~\cite{fary1948straight}, and Stein in 1951~\cite{stein1951convex} each independently proved that every planar graph admits a crossing-free, straight-line drawing.
These results are existential and do not lend themselves to algorithms that 
explicitly obtain straight-line, crossing-free drawings of planar graphs.
Later, algorithms that do create such drawings
were described by Tutte in 1963~\cite{tutte1963draw}, Chiba, Yamanouchi, and Nishizeki in 1984~\cite{chiba1984linear}, and Read in 1986~\cite{read1986new}.
The underlying difficulty of using these algorithms practically
is that the generated vertex positions are real numbers, rather than integers, and that the ratio between the longest and shortest edge length is exponential in the size of the graph.
Such drawings are difficult to read as the existence of very short edges implies vertices are placed very close to each other.

To draw planar graphs on a computer screen,
integer coordinates are preferable, as the screen itself is a grid.
Two algorithms for embedding planar graphs on integer grids 
were presented in 1990. 
The algorithm of de Fraysseix, Pach and Pollack~\cite{de1990draw} relies on canonical orders and
Schnyder's algorithm~\cite{schnyder1990embedding} uses a decomposition of the graph into three trees from which barycentric coordinates are computed. Both algorithms provide a straight-line  drawing on the
$\oh(n) \times \oh(n)$ grid,
where $n$ is the number of vertices.
Thus, the ratio between the longest and shortest edge is linear rather than exponential, avoiding the problem of unreadable regions where vertices are too close to each other. 

However, the edge lengths produced by these algorithms are not necessarily integral, leading to a natural question: 
\emph{Can planar graphs be drawn with straight-line edges on a (polynomial-size) grid so that all edge lengths are integers?}
The algorithms of de Fraysseix, Pach and Polack~\cite{de1990draw} and Schnyder~\cite{schnyder1990embedding} can (and indeed do) produce irrational edge lengths, but perhaps this can be avoided.
Harborth's conjecture~\cite{Harborth1987,kemnitz2001plane} asks whether every planar graph has a crossing-free realization with straight-line segments of integer lengths, referred to as integral F\'ary embeddings~\cite{chang2024harborth}.
A stronger version of Harborth's conjecture, called Kleber's conjecture, asks whether an integral F\'ary embedding exists where every vertex is located on the integer grid~\cite{kleber}.
We will refer to these drawings as \emph{truly integral F\'ary embeddings}.
\cref{fig:integer-examples} shows examples of truly integral F\'ary embeddings
on a small grid.

\begin{figure}[t]
    \centering
    \begin{minipage}[b]{0.3\textwidth}
    \begin{subfigure}{\textwidth}
    \centering
    \includegraphics[page=1]{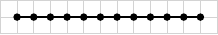}
    \subcaption{}
    \label{fig:example:1}    
    \end{subfigure}

    \begin{subfigure}{0.48\textwidth}
    \centering
    \includegraphics[page=3]{integer-examples.pdf}
    \subcaption{}
    \label{fig:example:2}    
    \end{subfigure}
    \hfill
    \begin{subfigure}{0.45\textwidth}
    \centering
    \includegraphics[page=4]{integer-examples.pdf}
    \subcaption{}
    \label{fig:example:3}    
    \end{subfigure}
    \end{minipage}
    \hfill
    \begin{minipage}[b]{0.23\textwidth}
    \centering
    \includegraphics[page=2]{integer-examples.pdf}
    \subcaption{\nolinenumbers}
    \label{fig:example:4}    
    \end{minipage}
    \hfill
    \begin{minipage}[b]{0.35\textwidth}
    \centering
    \includegraphics[page=5]{integer-examples.pdf}
    \subcaption{\nolinenumbers}
    \label{fig:example:5}    
    \end{minipage}
    \caption{Compact truly integral F\'ary embeddings of (a)~a path on $12$ vertices, (b)~a cycle on $12$ vertices, (c)~a cycle on $3$ vertices, (d)~a star on $13$ vertices and (e)~a full binary tree of depth $3$.}
    \label{fig:integer-examples}
\end{figure}
 
Sun~\cite{DBLP:conf/cccg/Sun11} uses rigidity matrices to prove that all planar (2, 3)-sparse graphs admit integral F\'ary embeddings.
In 2024, Chang and Sun~\cite{chang2024harborth} show this result also for all planar 4-regular graphs,
resolving issues from a previous attempt~\cite{DBLP:conf/cccg/Sun13}. 
A graph $G$ is ($a$, $b$)-sparse if every subgraph $G'$ of $G$ has at most $a |V(G')| - b$ edges.
For example, planar graphs are (3, 6)-sparse,
outerplanar graphs, series-parallel graphs, and 2-trees are (2, 3)-sparse,
while planar bipartite graphs are (2, 4)-sparse.
Geelen, Guo, and McKinnon~\cite{geelen2008straight} show that every planar max-degree-3 graph admits a truly integral F\'ary embedding.
Biedl~\cite{biedl2011drawing} observes that this technique can be applied to all planar (2, 1)-sparse graphs,
which include also all planar max-degree-4 graphs that are non-4-regular.
Independently, Benediktovich~\cite{DBLP:journals/dm/Benediktovich13} shows this result for max-degree-4 graphs that are non-4-regular and for planar 3-trees.
(%
For planar 3-trees, it was essentially also shown previously by Kemnitz and Harborth~\cite{kemnitz2001plane}.%
)
The latter set of results provide drawings where the vertices have rational coordinates, thus, they can be scaled to integers.
However, they make use of a theorem by Berry~\cite{berry} to position vertices on a rational point in an $\varepsilon$-disk around a real point.
The application of Berry's theorem does not provide guarantees on the distance to an integral point.
Thus, if 
the drawing were to be scaled to provide a truly integral F\'ary embedding, no bound on the area can be provided.

Biedl~\cite{biedl2011drawing} shows how to draw planar max-degree-3 graphs with integer edge lengths and integer coordinates, with the help of Kant's algorithm~\cite{DBLP:conf/wg/Kant92} for drawing 3-connected 3-regular graphs on the hexagonal grid.
Then, the drawing is skewed such that the three resulting slopes are part of Pythagorean triples (see below).
Non-3-connected parts are drawn independently, scaled, and pieced together.
This second step may require a large grid size. Chan~\cite{DBLP:journals/dcg/Chan20} provides several algorithms for drawing trees; one of them generates orthogonal straight-line drawings of binary trees in $n\cdot2^{\oh(\log^*n)}$ area. Since these drawings use only horizontal and vertical edges, the result 
also yields truly integral F\'ary embeddings.

\subparagraph*{Our contribution.}
In this paper, we study truly integral F\'ary embeddings of trees (\cref{sec:all-trees}) and cactus graphs (\cref{sec:cacti}).
Note that the previous results from the literature imply the existence of truly integral F\'ary embeddings
for trees and cactus graphs,
but they do not guarantee a small integer grid.
We provide algorithms producing such drawings on small polynomial-size integer grids.
To the best of our knowledge, the truly integral F\'ary embedding algorithms by Biedl~\cite{biedl2011drawing} for 3-connected 3-regular planar graphs on a quadratic-size grid and by Chan~\cite{DBLP:journals/dcg/Chan20} for binary trees on a linear-size grid are the only algorithms in the literature that provide polynomial upper bounds on the area.
Both of these results concern graphs of bounded maximum degree.
In contrast, we include trees and cactus graphs of unbounded maximum degree.
For an overview of our results, see \cref{tab:results}.

\begin{table}[t]
    \caption{Overview of results on truly integral F\'ary embeddings.}
    \label{tab:results}
    \centering  \scriptsize
    \sffamily \setlength{\tabcolsep}{6pt}
    \begin{tabular}{llclcl@{}}
  
        \toprule
        \multicolumn{2}{l}{\bfseries GRAPH CLASS} & \multicolumn{2}{c}{\bfseries UPPER BOUND} & \multicolumn{2}{c}{\bfseries LOWER BOUND} \\
        \midrule
        \medskip {\bfseries Stars} &  & $\oh(n)\times \oh(n)$ & Thm.~\ref{thm:stars-upper-bound} & $\Omega(n^2)$ & Thm.~\ref{thm:stars-lower-bound}  \\
        {\bfseries Trees} & \textit{---~~General}  & $\mathcal{O}\left(n \log n\right) \times \mathcal{O}\left(n \log n\right)$  & Thm.~\ref{thm:trees} & \multirow{2}{*}{$\Omega(n^2)$} &  \multirow{2}{*}{Thm.~\ref{thm:stars-lower-bound}}   \\
        & \textit{---~~Depth $d$}  & $\mathcal{O}\left(n\cdot d\right) \times \mathcal{O}\left(n\cdot d\right)$ & Thm.~\ref{thm:trees-bounded-depth} & &\\
       \smallskip & \textit{---~~Binary}  & $\mathcal \oh(A) \times  \frac{n}{A}\cdot2^{\mathcal{O}(\log^*n})$ & \multirow{2}{*}{\cite[Thm.~4.4]{DBLP:journals/dcg/Chan20}} & $\Omega(n)$  & trivial  \\
        \medskip &  & $\forall A \in [\log n, n/\log n]$ & &   &   \\
        {\bfseries Cacti}  & \textit{---~~General} & $\mathcal{O}\left(n^{2}\log n \right) \times \mathcal{O}\left(n^{2}\log{n}\right)$ & Thm.~\ref{thm:cactus}  &   \multirow{2}{*}{$\Omega(n^2)$} & \multirow{2}{*}{Thm.~\ref{thm:stars-lower-bound}}   \\
        & \textit{---~~Triangle-free} & $\mathcal{O}\left(n^{2}\right) \times \mathcal{O}\left(n^{2}\right)$ & Rem.~\ref{rmk:triangle-free} && \\
        \bottomrule\end{tabular}
\end{table}

Due to space constraints, some (details of) proofs are provided in the appendix; their statements are marked with \textcolor{yellow}{$\blacktriangle$}\textcolor{darkgray}{$\blacktriangledown$}. 
In the PDF, \textcolor{yellow}{$\blacktriangle$} links to the statement, \textcolor{darkgray}{$\blacktriangledown$} links to the proof.

\section{Preliminaries}
\label{sec:prelim}

We consider (straight-line) grid drawings of undirected graphs.
In a grid drawing, vertices are assigned to distinct grid points
and each edge is drawn as a straight-line segment between the grid points
of its two vertices.
For a vertex $v$, let $x_v$ and $y_v$ be the $x$- and $y$-coordinates of $v$, respectively.
The width $W$ of a grid drawing of a graph $G = (V, E)$
is $W=\max_{v \in V} x_v - \min_{v \in V} x_v$
and the height $H$ is $H=\max_{v \in V} y_v - \min_{v \in V} y_v$.
A \emph{tree} is a connected graph that does not contain any cycle.
A tree has $n-1$ edges, where $n$ is the number of vertices.
If we have a rooted tree~$T$ and a vertex~$v$ in $T$,
then we denote by~$T_v$ the subtree of~$T$ rooted at~$v$,
by $L(T_v)$ the leaves in~$T_v$,
and by~$n_v$ be the number of vertices in~$T_v$.
A collection of trees is called a \emph{forest}.
A \emph{cactus graph} is a $1$-vertex-connected graph such that each $2$-vertex-connected component is either a cycle or a single edge.
Observe that trees are cacti where every $2$-vertex-connected component is a single edge.
For $k\in \mathbb{N}$, we use $[k]$ to denote the set $\{1, \dots, k\}$.

\subparagraph*{On Pythagorean triples.}
A Pythagorean triple is a triple of integers $(x,y,z)$ that satisfy the equation $x^2 + y^2 = z^2$.
These triples are exactly the side lengths of right-angled triangles with integer side lengths. 
The existence of Pythagorean triples is among the oldest observations in mathematics, predating even Euclid.
Their algebraic definition and connection to geometry are likely some of the first non-trivial proofs shown in school.
In this section, we review a few properties of Pythagorean triples.
We provide explicit proofs for most of the statements in \cref{app:triples} of the appendix.

In the case of truly integral F\'ary embeddings,
we are restricted to using Pythagorean triples because the length of any line segment from a grid point $(x_1,y_1)$ to another grid point $(x_2,y_2)$
is given by the Euclidean distance $\sqrt{(x_2 - x_1)^2 + (y_2 - y_1)^2}$.
This length is integral only when the sum of square differences in the coordinates is a square.
In particular, if we enforce that the endpoints of the line segment are integer, then the difference~$x$ in x-dimension, the difference~$y$ in y-dimension, and the length~$z$ of the segment form a Pythagorean triple~$(x,y,z)$.

If a triple $(x, y, z)$ is a Pythagorean triple, then so is $kx, ky, kz$ for any positive integer~$k$.
A triple is called \textit{primitive} if there is no integer $k > 1$ that divides $x$, $y$, and $z$ evenly;
that is $x$, $y$, and $z$ are coprime.
We leverage that each primitive Pythagorean triple forms a triangle with a unique (clockwise) sequence of internal angles.
In particular, for a given (primitive) Pythagorean triple $(x,y,z)$, if a vertex $u$ is located at point $(x_u, y_u) \in \mathbb{Z}^2$ and has an edge to a not-yet-placed vertex $v$, then placing vertex $v$ at $(x_u+x,y_u+y)$ guarantees both integer coordinates for $v$ and integer edge length $z$ for the edge~$uv$.
If $u$ is adjacent to another vertex~$w$,
then using a different primitive Pythagorean triple guarantees,
due to the different angles, that the edges~$uv$ and~$uw$ do not overlap.
Our algorithms rely on the algorithmic generation of triples,
so we use a generator referred to as \emph{{Euclid}'s formula} (see, e.g.,~\cite{pythagorean}):

\begin{lemma}
\label{lem:euclid}
Let $i, j \in \mathbb{N}_{\geq 0}$. Then, the triple $(x,y,z)$ is a Pythagorean triple where
\begin{minipage}[b]{0.9\textwidth}
    \begin{equation}
\label{eq:euclid}    x=i^2-j^2, \quad y=2ij, \quad z=i^2+j^2.
\end{equation}
\end{minipage}
\end{lemma}

\begin{proof}
    Since $i$ and $j$ are integer, so are $x$, $y$ and $z$. Moreover, $x^2+y^2=(i^4-2i^2j^2+j^4)+(4i^2j^2)=(i^4+2i^2j^2+j^4)=z^2$.
\end{proof}

We next assure that we can generate all primitive Pythagorean triples using {Euclid}'s formula.
The following theorem has been shown for instance by Sierpi\'nski~\cite[Ch. 2.5]{pythagorean}.

\begin{theorem}[\cite{pythagorean}]
\label{thm:allprimitive}
Let $P=(x,y,z)$ be a primitive Pythagorean triple such that $y$ is even.
Then, there is exactly one pair of co-prime integers $i, j \in \mathbb{N}_{\geq 0}$ with $i > j$ such that $P$ is generated using Equation~\eqref{eq:euclid}.
\end{theorem}

One may now wonder what happens if $y$ is odd.
According to \cref{thm:allprimitive}, the following generator generates all primitive Pythagorean triples where $x$ is even for $i, j \in \mathbb{N}_{\geq 0}$:
\begin{equation}
    \label{eq:euclid2}\tag{1'}    x=2ij, \quad y=i^2-j^2,\quad z=i^2+j^2
\end{equation}

The following corollary establishes that \eqref{eq:euclid} and \eqref{eq:euclid2} are sufficient to generate all primitive Pythagorean triples. While this may be (not quite straight-forwardly) concluded by observing that $x$ and $y$ must be co-prime~\cite[Ch. 1.5]{pythagorean},
we make use of the following special case:

\begin{lemma}
\label{lem:oneOddOneEven}
    In any primitive Pythagorean triple $(x, y, z)$, exactly one of $x$ and $y$ is odd or $x = y = z = 0$.
\end{lemma}

\begin{proof}
    First, assume for a contradiction that both $x$ and $y$ are even, i.e, $x=2x'$ and $y=2y'$ for $x',y'\in \mathbb{N}_{\geq 0}$. Then, $z$ is also even, i.e., $z=2z'$ for $z' \in \mathbb{N}_{\geq 0}$. But then $(x',y',z')=\left(\frac{x}{2},\frac{y}{2},\frac{z}{2}\right)$ is a Pythagorean triple; a contradiction to the fact that $(x,y,z)$ is primitive
    (unless $x = y = z = 0$).
    
    Second, assume for a contradiction that both $x$ and $y$ are odd, i.e., $x=2x'+1$ and $y=2y'+1$ for $x', y' \in \mathbb{N}_{\geq 0}$. Then, $x^2$ and $y^2$ are also odd as they are the product of two odd numbers and $z^2$ must be even as it is the sum of two odd numbers. Hence, also $z$ must be even, i.e., $z=2z'$ for $z' \in \mathbb{N}_{\geq 0}$. But then:
    \begin{eqnarray}
        \notag  x^2+y^2&=&z^2\\
        \notag \Leftrightarrow  (2x'+1)^2+(2y'+1)^2&=&(2z')^2\\
        \notag \Leftrightarrow  4x'^2+4x'+1+4y'^2+4y'+1&=&4z'^2\\
        \notag \Leftrightarrow  4(x'^2+x'+y'^2+y')+2&=&4z'^2\\
        \Leftrightarrow (x'^2+x'+y'^2+y')+\frac{1}{2}&=&z'^2\label{eq:oddproof}
    \end{eqnarray}
    
    The left-hand side of \eqref{eq:oddproof} contains the term $(x'^2+x'+y'^2+y')$ which is a natural number as $x', y' \in \mathbb{N}$. Thus, $(x'^2+x'+y'^2+y') + \frac{1}{2}\notin \mathbb{N}$. On the other hand, $z'^2\in\mathbb{N}$; a contradiction.
\end{proof}

\cref{thm:allprimitive,lem:oneOddOneEven} immediately imply the following:

\begin{corollary}\label{cor:allprimitive}
Let $P=(x,y,z)$ be any primitive Pythagorean triple.
Then, there is exactly one pair of co-prime integers $i, j \in \mathbb{N}_{\geq 0}$
with $i > j$ such that $P$ is generated using Equation~\eqref{eq:euclid} or~\eqref{eq:euclid2}.
\end{corollary}

Equations~\eqref{eq:euclid} and~\eqref{eq:euclid2} provide us with a possibility to generate the Pythagorean triples.
For our algorithmic results, we are 
interested in the asymptotic growth in the numbers of the Pythagorean triples as these will influence the required area and running time.
A preliminary review of the growth of Pythagorean triples
can for instance be found in~\cite[Ch. 2.7, 3]{pythagorean}.
Since we are not aware of literature discussing this explicitly,
we will provide a complete proof
based on the asymptotic behavior of the Farey sequence.
The Farey sequence of order $n$ is the increasingly ordered sequence
of completely reduced fractions that have denominators less than or equal to $n$.

In this context, we define an ordering $\prec_{\text{prim}}$ of the primitive Pythagorean triples.
From now on, we exclude the trivial Pythagorean triple $(0, 0, 0)$,
which can be generated by $i = j = 0$.
Let $P=(x,y,z)$ and $P'=(x',y',z')$ be two primitive Pythagorean triples with $P\neq P'$.
By \cref{cor:allprimitive}, there are integers $i, j \in \mathbb{N}_{\ge0}$
such that $P$ is generated when inserting $i$ and $j$ into Equation~\eqref{eq:euclid} or~\eqref{eq:euclid2} and there are integers $i', j' \in \mathbb{N}_{\ge0}$ such that $P'$ is generated  when inserting $i'$ and $j'$ into Equation~\eqref{eq:euclid} or~\eqref{eq:euclid2}.
Note that $i \ge j$ as $z>0$. We now say that $P \prec_{\text{prim}} P'$ if and only if one of the following applies:
\begin{enumerate*}[label=\textcolor{lipicsGray}{\textsf{\textbf{\arabic*.}}}]
    \item $i < i'$,
    \item $i=i'$ and $j < j'$, or
    \item $j=j'$, $i=i'$ and $P$ is generated using Equation~\eqref{eq:euclid} (i.e., $P'$ is generated using~\eqref{eq:euclid2}).
\end{enumerate*}

Observe that $\prec_{\text{prim}}$ defines a total ordering of the primitive Pythagorean triples.
Hence, in the following, let $\mathcal{P}=(P_1,P_2,\ldots)$ denote the set of primitive Pythagorean triples ordered with respect to~$\prec_{\text{prim}}$.
Moreover, we say that $P_k$ is the \emph{$k$-th} primitive Pythagorean triple and that $\mathcal{P}_k=(P_1,\ldots,P_k) \subset \mathcal{P}$ is the set of the \emph{first} $k$ primitive Pythagorean triples.
We can bound the numerical values $x_k$, $y_k$ and $z_k$ in $P_k=(x_k,y_k,z_k)$:

\begin{lemma}
    \label{lem:triple_size}
    Let $P_k=(x_k,y_k,z_k) \in \mathcal{P}$ be the $k$-th primitive Pythagorean triple obtained by \eqref{eq:euclid}~or%
    ~\eqref{eq:euclid2} using $i$,~$j$.
    Then, $i, j < \frac{\pi}{\sqrt{6}}\sqrt{k}\in\oh\left(\sqrt{k}\right)$ and $x_k,y_k,z_k < \frac{\pi^2}{{3}}{k} \in \oh(k)$.
\end{lemma}

\begin{proof}
    First observe that a scaled Pythagorean triple
    is generated by the same but scaled set of input numbers in the generator formulas.
    We use the following claim to state and prove this observation.
    \begin{claim}\label{clm:primitive_property}
        Let $i,j,a,b,s \in \mathbb{N}_{+}$ such that $i > j$, $a > b$ and
        \begin{eqnarray}
            \label{eq:linear:1} i^2-j^2& = &s(a^2 - b^2) \\
            \label{eq:linear:2} 2ij &= &s(2ab) \\
            \label{eq:linear:3} i^2+j^2 &= &s(a^2 + b^2).
        \end{eqnarray}    \text{Then,} \begin{minipage}[b]{0.4\textwidth}\begin{equation}
                \label{eq:sameFraction}\frac{j}{i}=\frac{b}{a}.
        \end{equation}\end{minipage}
    \end{claim}
    
    \begin{claimproof}
        We begin by calculating the sum of \eqref{eq:linear:1} and \eqref{eq:linear:3}:
        \begin{equation}
            \label{eq:linear:4} i^2-j^2 + i^2+j^2  = s(a^2 - b^2) + s(a^2 + b^2) \Leftrightarrow 2i^2 = 2sa^2 \Leftrightarrow s = \frac{i^2}{a^2}
        \end{equation}
        
        Next, we compute  the  difference between \eqref{eq:linear:3} and \eqref{eq:linear:1}:
        \begin{equation}
            \label{eq:linear:5} i^2+j^2 - (i^2 - j^2 ) = s(a^2 + b^2) - s(a^2 - b^2) \Leftrightarrow 2j^2 = 2sb^2 \Leftrightarrow s = \frac{j^2}{b^2}
        \end{equation}
        
        Both \eqref{eq:linear:4} and \eqref{eq:linear:5} provide a description of $s$, which we can set equal:
        
        \begin{equation}
            \label{eq:linear:6}
            \frac{i^2}{a^2}=s=\frac{j^2}{b^2} \Leftrightarrow \frac{j^2}{i^2}=\frac{b^2}{a^2} \Leftrightarrow \left(\frac{j}{i}\right)^2=\left(\frac{b}{a}\right)^2
        \end{equation}
        
        Since $i,j,a,b\geq 1$, Equation~\eqref{eq:linear:6} implies \eqref{eq:sameFraction} which concludes the proof.
    \end{claimproof}

    According to \cref{clm:primitive_property}, we can obtain two primitive Pythagorean triples for each irreducible fraction $\frac{j}{i}$ using numerator $j$ and denominator $i$ as input for Equation~\eqref{eq:euclid} and~\eqref{eq:euclid2}, respectively.   
    The sequence (sorted by value) of irreducible fractions $\frac{j}{i}$  where $0\leq j \le i \leq t$ is known as the Farey sequence $\mathcal{F}_t$ of order $t$~\cite[Chapter 4.5]{farey}.  It is known that 
    \begin{equation*}
        |\mathcal{F}_t|=\frac{3}{\pi^2}t^2 + 1 + f(t) \quad \text{where } f(t) \in \oh(t \log t) \text{ and } f(t)\geq 0.\quad\cite[\text{Chapter 9.3}]{farey}
    \end{equation*}
    In particular, for $t= \lfloor \frac{\pi}{\sqrt{6}}\sqrt{k} \rfloor$, we obtain:
    \begin{equation*}
        |\mathcal{F}_t| > \frac{3}{\pi^2} \left(\frac{\pi}{\sqrt{6}}\sqrt{k}\right)^2 + 1
        = \frac{3\pi^2}{\pi^26}k  + 1
        = \frac k2 + 1
    \end{equation*}
    
    In our case, this means that, if $0 \le j \le i < \frac{\pi}{\sqrt{6}}\sqrt{k}$,
    $i$ and $j$ constitute at least $\frac k2 + 1$ fractions
    of the Farey sequence.
    If we exclude the fraction $0$ (where $i = j = 0$),
    we can use $i$ and $j$ of the remaining $\frac k2$ fractions to generate
    $\frac k2$ primitive Pythagorean triples with Equation~\eqref{eq:euclid} and
    $\frac k2$ primitive Pythagorean triples with Equation~\eqref{eq:euclid2}.
    Therefore, in the $k$-th primitive Pythagorean triple,
    it holds that $x_k,y_k,z_k < \frac{\pi^2}{3}k$.
\end{proof}

Given \cref{lem:triple_size}, it is now straight-forward to compute the first $k$ Pythagorean triples:

\begin{lemma}
\label{lem:compute_triples}
The first $k$ Pythagorean triples $\mathcal{P}_k$ can be computed in $\oh(k)$ time.
\end{lemma}

\begin{proof}
    According to \cref{lem:triple_size}, there is a value $t=\oh\left(\sqrt{k}\right)$,
    so that each $P \in \p_k$ can be generated by inserting values $i$ and $j$ into Equation~\eqref{eq:euclid} or~\eqref{eq:euclid2} where $i,j \leq t$.
    Recall form the proof of \cref{lem:triple_size}
    that there is a 1--to--1 correspondence between
    the Farey sequence $\mathcal{F}_t$ of order $t$
    and the first $k$ (or more) primitive Pythagorean triples.
    Hence, if we can generate all fractions in $\mathcal{F}_t$,
    we can construct, in the same time bound, $\p_k$.
    It is well known that the next term in $\mathcal{F}_t$
    can be obtained in constant time;
    over all fractions, this is in $\oh(t^2) \subseteq \oh(k)$ time.
    We directly compute the corresponding primitive Pythagorean triple~$(x, y, z)$ and
    store it in a list~$L$.
    If needed, we can sort $L$ in linear time w.r.t.~$k$ since,
    for each entry $(x, y, z)$, we have $x, y, z \in \oh(k)$.
\end{proof}

We use primitive Pythagorean triples $(x,y,z)$ to draw edges -- the two endpoints of an edge $e$ will be $x$ units horizontally and $y$ units vertically apart while $e$ has length $z$.
To maintain planarity, we sort the edges by slope.
The \emph{angle} $\alpha(P)$ of a Pythagorean triple $P=(x,y,z)$ is equal to $\arctan(y/x)$, i.e., it is the angle in the right-angled triangle spanned by edges of lengths $x$, $y$, and $z$ that occurs between the edges of lengths $x$ and $z$.
For Pythagorean triples $P$ and $P'$, we say that $P \prec_\text{angle} P'$ if and only if $\alpha(P)<\alpha(P')$.
Observe that, due to the definition of the $\arctan$,
all angles that correspond to primitive Pythagorean triples of~$\mathcal{P}$ are distinct
and, hence, $\prec_\text{angle}$ defines a total ordering of $\mathcal{P}$.
In the following, we denote by $\mathcal{P}_k^\circ=(P_1 ^\circ,\ldots,P_k^\circ)$ the \emph{angle-sorted} first $k$ primitive Pythagorean triples,
that is, the permutation of $\mathcal{P}_k$ so that for $s, t \in [k]$ and $s < t$, we have $P_s^\circ \prec_\text{angle} P_t^\circ$.

For our tree and cactus algorithms,
we use a specific set~$\q_{n,a}$ of (not-necessarily primitive)
Pythagorean triples 
where $n$ describes the size of the set and $a$ is a scaling factor.
We define this set next.
Among the first $4n$ primitive Pythagorean triples~$\p_{4n}$,
let $\p_{4n}^\eqref{eq:euclid}$ be the $2n$ primitive Pythagorean triples that are generated by \eqref{eq:euclid}.
From $\p_{4n}^\eqref{eq:euclid}$,
take the $n$ primitive Pythagorean triples that correspond to the smallest angles and scale each such triple~$P = (x, y, z)$ with the smallest factor $c \in \mathbb{N}$
such that $c \cdot x > \frac{4a\pi^2}{3}n$.
The resulting set of Pythagorean triples is~$\q_{n,a}$.
For the case that $a = 1$, we also write just~$\q_n$.
Similar to~$\p_n^\circ$, we use $\q_{n,a}^\circ$
to denote the permutation of $\q_{n,a}$ ordered by angle.

Next we show two properties about the Pythagorean triples in~$\q_{n,a}$,
which will be useful to ensure planarity of the drawings created by the algorithms.
Namely, we show that they all have a ``small'' angle and,
if we draw all hypotenuses of these Pythagorean triples with a common bottom point,
their top points have a guaranteed vertical distance; see \cref{fig:lemma-vertical-gap}.

\begin{figure}
    \centering
    \includegraphics{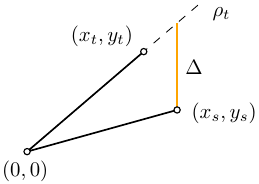}
    \caption{Vertical gap w.r.t.\ two Pythagorean triples $(x_s, y_s, z_s)$ and $(x_t, y_t, z_t)$
        as in \cref{lem:vertical-gap}.}
    \label{fig:lemma-vertical-gap}
\end{figure}

\begin{lemma}
\label{lem:lower-half-triples}
For a Pythagorean triple $P \in \q_{n, a}$ being generated using $i, j \in \mathbb{N}$ in Equation~\eqref{eq:euclid} (and being scaled),
we have $\alpha(P) = 2\arctan\left(\frac{j}{i}\right) < 2 \arctan\left(\frac12\right) = \arctan(4/3)
\approx 53.1^\circ$.
\end{lemma}

\begin{proof}
    For a Pythagorean triple $P = (c \cdot (i^2-j^2), c \cdot 2ij, c \cdot (i^2+j^2))$ generated by \eqref{eq:euclid} and scaled by~$c \in \mathbb{N}$,
    we have \[\alpha(P) = \arctan\left(\frac{c \cdot 2ij}{c \cdot (i^2-j^2)}\right) = \arctan\left(\frac{2ij}{i^2-j^2}\right) = 2\arctan\left(\frac{j}{i}\right),\]
    which one can see by using the identity
    $\frac{2\tan(x)}{1-\tan^2(x)} = \tan(2x)$, substituting $x = \arctan\left(\frac{j}{i}\right)$ and applying $\arctan()$.
    In the one-to-one correspondence between irreducible fractions (representing $i$ and $j$) and primitive Pythagorean triples in \eqref{eq:euclid}, smaller fractions generate triples with a smaller angle.
    The set of irreducible fractions in the unit interval is symmetric around $\frac{1}{2}$, so the triples of $\q_{n, a}$ are those (after scaling) generated by $\frac{j}{i} < \frac{1}{2}$.
    With the equation above, we get the claimed angle.
\end{proof}

\begin{lemma}
\label{lem:vertical-gap}
Let $P_1, \dots, P_n$ be the (ordered) Pythagorean triples from $\q_{n, a}^\circ$
and let each such Pythagorean triple $P_i=(x_i, y_i, z_i)$ be represented by a line segment
with one endpoint at $(0, 0)$ and with the other endpoint at $(x_i, y_i)$.
Then, for each $s, t \in [n]$ with $s < t$,
we have that the vertical distance $\Delta$ between the top point $(x_s, y_s)$ of the segment of~$P_s$
and the ray~$\rho_t$ originating in $(0, 0)$ and containing the segment of $P_t$ is $\Delta > a (t - s)$.
\end{lemma}

\begin{proof}
    Consider the angle between two consecutive
    Pythagorean triples $P, P' \in \q_{n,a}^\circ$ with $\alpha(P) > \alpha(P')$,
    which are scaled versions of primitive Pythagorean triples generated by Equation~\eqref{eq:euclid}
    using $i, j \in \mathbb{N}$, and $i', j' \in \mathbb{N}$, respectively.
    We will use the following well-known equation and the following upper bound for $0 \le y < x \le 1$.
    \begin{equation}
        \label{eq:arctan}
        \arctan(x) - \arctan(y) = \arctan \left( \frac{x-y}{1 + xy} \right)
        \ge \arctan \left( \frac{x-y}{2} \right)
    \end{equation}
    Since $\arctan(x) \ge \frac{x}{2}$ holds for $0 \le x \le \pi/2$ and due to \cref{lem:lower-half-triples}, we have
    \begin{equation}
        \label{eq:diff-alphaPP'}
        \alpha(P) - \alpha(P')
        = 2 \left(\arctan\left(\frac{j}{i}\right) - \arctan\left(\frac{j'}{i'}\right)\right)
        \overset{\eqref{eq:arctan}}{\ge} 2 \arctan\left(\frac{\frac{j}{i} -\frac{j'}{i'}}{2}\right) \ge \frac{i'j-ij'}{2ii'}.
    \end{equation}
    We next use the facts that
    \begin{equation}
        \label{eq:i'j-ij'ge1}
        \frac{j}{i} - \frac{j'}{i'} > 0 ~~\Leftrightarrow~~ i'j - ij' > 0 ~~\Rightarrow~~ i'j - ij' \ge 1 ~~\textrm{(since } i, j, i', j' \textrm{ are integers)}
    \end{equation}
    and that $i,i' < \frac{\pi \sqrt{2}}{\sqrt{3}}\sqrt{n}$ (see \cref{lem:triple_size}) to lower-bound $\alpha(P) - \alpha(P')$:
    \begin{equation}
        \label{eq:diff-alphaPP'2}
        \alpha(P) - \alpha(P') \overset{\eqref{eq:diff-alphaPP'}}{\ge} \frac{i'j-ij'}{2ii'}
        \overset{\eqref{eq:i'j-ij'ge1}}{\ge} \frac{1}{2ii'} > \frac{3}{4\pi^2 n}.
    \end{equation}
    We can easily generalize \eqref{eq:diff-alphaPP'2} to not-necessarily consecutive
    Pythagorean triples $P_s, P_t \in \q_{n,a}^\circ$ with $\alpha(P_s) < \alpha(P_t)$
    by adding up the angle difference guaranteed by \eqref{eq:diff-alphaPP'2}
    for each consecutive Pythagorean triple in the subsequence of $\q_{n,a}^\circ$ from $P_s$ to $P_t$:
    \begin{equation}
        \label{eq:diff-alphaPP'3}
        \alpha(P_t) - \alpha(P_s) > (t-s) \cdot \frac{3}{4\pi^2 n}.
    \end{equation}
    Hence, the vertical distance $\Delta$ from the point $(x_s, y_s)$ to $\rho_{t}$ is at least
    \begin{align*}
        \Delta &\ge \left(\tan(\alpha(P_t)) - \tan(\alpha(P_s)) \right) \cdot \frac{4a\pi^2n}{3} \\
        & \ge \tan\left(\alpha(P_t)-\alpha(P_s)\right) \cdot \frac{4a\pi^2n}{3}\\
        & \overset{\eqref{eq:diff-alphaPP'3}}{>} \tan\left((t - s) \cdot \frac{3}{4\pi^2 n}\right) \cdot \frac{4a\pi^2n}{3}\\
        & \ge (t - s) \cdot \frac{3}{4\pi^2 n}\cdot \frac{4a\pi^2n}{3}\\
        & = a (t - s),
    \end{align*}
    using that $\tan(x)-\tan(y) \ge \tan(x-y)$ and $\tan(x) \ge x$ for $0 \le y<x\le \pi/2$.
\end{proof}

\section{Truly Integral F\'ary Embeddings of Trees in Polynomial Area}
\label{sec:all-trees}

We show that trees of arbitrary degree admit truly integral F\'ary embeddings in polynomial area. 
We first investigate stars (\cref{thm:stars-upper-bound,thm:stars-lower-bound}) and extend to a simple algorithm for trees (\cref{thm:trees-bounded-depth}). Then, we
give a better drawing algorithm for trees
using an $\oh(n \log n) \times \oh(n \log n)$ size grid~(\cref{thm:trees}).
In \cref{sec:cacti}, we generalize to cacti.

\subsection{Stars}
\label{sec:stars}

\begin{figure}
    \centering
    \includegraphics{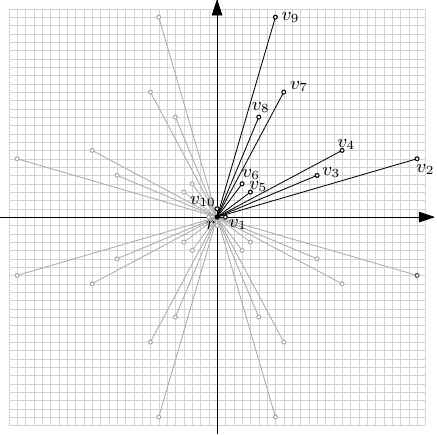}
    \caption{Drawing of a star $T$ with $36$ leaves created with our algorithm from \cref{thm:stars-upper-bound}.}
    \label{fig:star}
\end{figure}

\begin{theorem}
\label{thm:stars-upper-bound}
    Let $T = (V, E)$ be a star with $n$ vertices.
    There is a truly integral F\'ary~embedding on
    a grid of size $\left\lfloor \frac{\pi^2}{6}n+12 \right\rfloor \times \left\lfloor \frac{\pi^2}{6}n+12 \right\rfloor \approx 1.645 n \times 1.645 n$,
    which can be found in $\oh(n)$ time.
\end{theorem}

\begin{proof}
    We describe how to draw $T$ such that every vertex lies on a grid point
    and every edge has integer length.
    The realization as a linear-time algorithm is straightforward;
    determining the $\oh(n)$ Pythagorean triples can also be done
    in $\oh(n)$ time according to \cref{lem:compute_triples}.
    For an illustration, see \cref{fig:star}.
    
    Place the central vertex~$r$ at $(0, 0)$.
    Label the leaves (arbitrarily or in a given rotation order) $v_1, v_2, \dots, v_n$.
    For $i \in [\lceil n / 4 \rceil + 1]$,
    place $v_i$ onto $(x_i, y_i)$ where $(x_i, y_i, z_i)$
    is the $i$-th primitive Pythagorean triple from~$\p^\circ_{\lceil n / 4 \rceil + 1}$.
    Clearly, each edge $rv_i$ has an integral length $z_i$
    and no pair of edges overlap.
    So far, we have drawn about one fourth of the edges in the top right quadrant.
    Note that the first and the last Pythagorean triple from~$\p^\circ_{\lceil n / 4 \rceil + 1}$
    correspond to horizontal and vertical line segments.
    Similarly, repeat this process for the remaining three fourths of the leaves~-- for each of them in one of the three remaining quadrants
    with the same set of primitive Pythagorean triples~--
    except for horizontal and vertical directions that are already occupied.
    
    The resulting drawing is a truly integral F\'ary embedding.
    We next analyze the grid size.
    By \cref{lem:triple_size}, the numbers of the $k$-th Pythagorean triple $(x_k, y_k, z_k)$
    are bounded by $x_k,y_k \leq \pi^2k/3$.
    Since in our case $k \le \lceil n / 4 \rceil + 1$,
    the maximum x- and y-coordinate is at most $\frac{\pi^2}{{12}}{n}+6$ and
    the minimum x- and y-coordinate is at least $-\frac{\pi^2}{{12}}{n}+6$.
    This yields a grid of size at most the following
    (which we can round down to the next integer):
    \begin{equation*}
        \left(2 \cdot \left(\frac{\pi^2}{{12}}{n}+6\right)\right) \times \left(2 \cdot \left(\frac{\pi^2}{{12}}{n}+6\right)\right) =
        \left(\frac{\pi^2 n}{6}+12\right) \times \left(\frac{\pi^2 n}{6}+12\right). \qedhere
    \end{equation*}
\end{proof}

Next we prove an asymptotically matching lower bound of $\Omega(n^2)$ on the grid size required
for truly integral F\'ary embeddings of stars.

\begin{theorem}
\label{thm:stars-lower-bound}
    Any truly integral F\'ary embedding of a star with $n$ vertices needs
    a grid of size $\Omega(n^2)$.
\end{theorem}

\begin{proof}
    In any drawing of a star, all neighbors of the central vertex
    must be placed on integer grid points and the edge
    must have integer length.
    This requires at least $\lceil\frac{n}{4}\rceil$ distinct primitive Pythagorean triples.
    (W.l.o.g., we can assume that these Pythagorean triples
    are primitive because otherwise,
    we can shorten the edge to correspond
    to a primitive Pythagorean triple without increasing the grid size.)
    
    By \cref{cor:allprimitive}, we know that
    the generator formulas~\eqref{eq:euclid} and~\eqref{eq:euclid2}
    generate all primitive Pythagorean triples.
    Let $i$ and $j$ be the two input numbers of these generator formulas with $i > j$.
    To generate $\lceil\frac{n}{4}\rceil$ primitive Pythagorean triples,
    we need $\lceil\frac{n}{4}\rceil$ distinct pairs of $i$ and~$j$.
    Let this set of pairs be $P$.
    
    We will show that, for at least one pair $(i, j) \in P$,
    the product of the two generator functions
    is in $\Omega(n^2)$.
    We denote this product by $f(i,j)$.
    Namely, we show that there is an $(i, j) \in P$ such that
    \begin{equation*}
        f(i,j) = 2ij \cdot \left(i^2-j^2\right) = 2ij(i-j)(i+j) \ge \frac{n^2}{64}.
    \end{equation*}
    To this end, we will count for how many pairs of $(i, j)$,
    we can have $f(i,j) < n^2/64$.
    
    First consider all pairs where $j < i/2$.
    Since $i-j \ge i/2$ and $i + j \ge i$,
    we can write:
    \begin{equation}
        \label{eq:small-j}
        \frac{n^2}{64} > f(i,j) = 2ij(i-j)(i+j) \ge i^3j
    \end{equation}
    Resolving \eqref{eq:small-j} for~$j$ yields $j < n^2/(64 i^3)$.
    As we also know $j < i/2$, we have
    \begin{equation}
        \label{eq:bound-j}
        j < \min \left\{ \frac i2, \frac{n^2}{64 i^3} \right\} .
    \end{equation}
    Now we sum over all $i$ to count the number~$N_1$ of such pairs $(i, j)$ using \eqref{eq:bound-j}:
    \begin{equation}
        \label{eq:N1}
        N_1 = \sum_{i=1}^{\infty} \left\lfloor \min \left\{ \frac i2, \frac{n^2}{64 i^3} \right\} \right\rfloor
        \le \sum_{i=1}^{\lfloor \sqrt{n}/2 \rfloor - 1} \frac i2 + \sum_{i=\lfloor \sqrt{n}/2 \rfloor}^{\infty} \frac{n^2}{64 i^3}
        \le \frac{\frac n4 - \frac{\sqrt{n}}{2}}{4} + \frac{n^2}{64} \cdot \frac{3}{n}
        < \frac{5n}{64}
    \end{equation}
    
    Now consider all pairs where $j \ge i/2$,
    and re-write $i$ as $j + d$ for an integer $d \ge 1$:
    \begin{equation*}
        f(j+d,j) = 2(j+d)j \cdot \left((j+d)^2-j^2\right) = 2(j+d)jd(2j+d)
    \end{equation*}
    We can bound $(j+d) < j$ and $(2j+d) < 2j$, and
    we are interested in the number of pairs of a $j$ and a suitable $d$ that satisfy
    \begin{equation}
        \label{eq:boundfdj}
        \frac{n^2}{64} > f(j+d, d) = 2(j+d)jd(2j+d) > 4j^3d.
    \end{equation}
    Resolving \eqref{eq:boundfdj} for~$d$ yields $d < n^2/(256 j^3)$.
    As we also know $d \le j$, we have
    \begin{equation}
        \label{eq:bound-d}
        d \le \min \left\{ j, \frac{n^2}{256 j^3} \right\} .
    \end{equation}
    Now we sum over all $j$ to count the number~$N_2$ of such pairs $(j+d, j)$ using \eqref{eq:bound-d}:
    \begin{equation}
        \label{eq:N2}
        N_2 = \sum_{j=1}^{\infty} \left\lfloor \min \left\{ j, \frac{n^2}{256 j^3} \right\} \right\rfloor
        \le \sum_{j=1}^{\lfloor \sqrt{n}/4 \rfloor - 1} j + \sum_{j=\lfloor \sqrt{n}/4 \rfloor}^{\infty} \frac{n^2}{256 j^3}
        \le \frac{\frac{n}{16} - \frac{\sqrt{n}}{4}}{2} + \frac{n^2}{256} \cdot \frac{24}{n}
        < \frac{n}{8}
    \end{equation}
    
    Adding up $N_1$ and $N_2$ from \eqref{eq:N1} and \eqref{eq:N2},
    we see that we have less than $\frac{13n}{64}$ (which is smaller than $\lceil\frac n4\rceil$) pairs $(i,j)$
    with $f(i,j) < n^2/64$, which shows the claim.
\end{proof}

\subsection{Trees of Bounded Depth}
\label{sec:trees-bounded-depth}

The algorithm in \cref{thm:stars-upper-bound} for stars,
gives us a straight-forward algorithm for trees
(see \cref{fig:tree}) with $\ell$ leaves:
draw a star at the root and go into recursion on each child.
To avoid overlaps, assign as many distinct Pythagorean triples from $\p^\circ_{\ell}$
to subtrees as needed there.
We obtain the following result, which is in particular interesting for trees of low depth (depth in $\oh(\log n)$; for depth in $O(1)$, it is worst-case optimal by \cref{thm:stars-lower-bound}).
We give details in the appendix.

\begin{figure}
    \centering
    \includegraphics{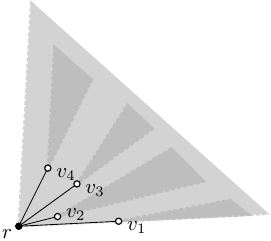}
    \caption{Schematization of our recursive algorithm from \cref{thm:trees-bounded-depth}:
        Vertex $r$ is placed at $(0,0)$ and its subtree is drawn inside the light-gray cone.
        Each child $v_i$ is assigned an exclusive set of slopes resulting in the non-overlapping gray cones centered at the child components.
        The edge between $r$ and $v_i$ is drawn using the first triple assigned to $v_i$.}
    \label{fig:tree}
\end{figure}

\begin{theorem}
\label{thm:trees-bounded-depth}
    Let $T = (V, E)$ be a (rooted) tree, which has $n$ vertices, $\ell$ leaves, and depth~$d$.
    There is a truly integral F\'ary embedding on a grid of size $\frac{\pi^2}{3} \ell d \times \frac{\pi^2}{3} \ell d \subseteq \oh(n^2) \times \oh(n^2)$,
    which can be found in $\oh(n)$ time.
\end{theorem}

\begin{proof}
    If no root is given explicitly,
    we can choose an arbitrary vertex as root.
    Let this be~$r$.
    We start by describing a recursive algorithm for drawing~$T$, which we analyze afterwards.
    
    \subparagraph*{Algorithm.}
    Similar to stars, we place~$r$ at $(0, 0)$ and
    we use $\ell$ primitive Pythagorean triples for drawing $T$.
    Assign the first child~$v_1$ of~$r$ the first $|\leaves(T_{v_1})|$
    primitive Pythagorean triples from~$\p^\circ_{\ell}$,
    assign the second child~$v_2$ of~$r$ the next $|\leaves(T_{v_2})|$
    primitive Pythagorean triples from~$\p^\circ_{\ell}$, etc.
    Draw each child~$v$ of~$r$ at the coordinates $(x, y)$
    where $(x, y, z)$ is the first primitive Pythagorean triple among $\p(v)$,
    where $\p(v)$ is the set of primitive Pythagorean triples assigned to~$v$, sorted as in~$\p^\circ_{\ell}$.
    For each child~$v$ of~$r$, for which~$v$ is not a leaf of~$T$,
    we draw the subtree~$T_v$ recursively:
    in the recursive call, we assume that the coordinate system is centered at the position of~$v$ and we use exactly the primitive Pythagorean triples assigned to~$v$; see \cref{fig:tree}.
    
    \subparagraph{Correctness.} We next show that for every vertex~$v$ of~$T$, the following invariants hold:
    \begin{enumerate}
        \item\label{inv:plane2} The drawing of~$T_v$ is a truly integral F\'ary embedding.
        \item\label{inv:cone2} The drawing of $T_v$ lies in a cone
        that is rooted at~$v$ and bounded by two rays emanating at~$v$
        and having slopes $y_1/x_1$ and $y_{|\leaves(T_v)|} / x_{|\leaves(T_v)|}$,
        where the triples $(x_1, y_1, z_1)$ and $(x_{|\leaves(T_v)|}, y_{|\leaves(T_v)|}, z_{|\leaves(T_v)|})$
        denote the first and the last Pythagorean triple of~$\p(v)$, respectively. See also the gray cones in \cref{fig:tree} for an illustration.
        \item\label{inv:grid2} In the drawing of~$T_v$, the Euclidean distance between
        $v$ and any vertex of~$T_v$ is at most $d(T_v) \cdot \frac{\pi^2}{3}\cdot  \ell$.
    \end{enumerate}
    
    \subparagraph{Base case: a leaf.} We start with the base case that $v$ is a leaf.
    The drawing of~$T_v$ contains only $v$ being placed at~$(0, 0)$.
    Trivially, \cref{inv:plane2,inv:cone2,inv:grid2} hold.
    
    \subparagraph{Recursion: combining child components of a vertex $u$.} We consider the combination step where $u$ is the root
    of~$T_u$ and $u$ has children $v_1, \dots, v_k$ for some $k \in \{1,\ldots,n-1\}$.
    We assume that $T_{v_1}, \dots, T_{v_k}$
    are drawn according to \cref{inv:plane2,inv:cone2,inv:grid2}.
    We place $u$ at $(0, 0)$ and, for each $i \in \{1,\ldots,k\}$,
    we place the drawing of~$T_{v_i}$ such that $v_i$ lies on $(x, y)$
    where $(x, y, z)$ is the first primitive Pythagorean triple of~$\p(v_i)$.
    This is a translation of~$T_{v_i}$ from a grid point to another grid point; see also the child components of $r$ in \cref{fig:tree}.
    All vertices of~$T_{v_i}$ as well as $u$ lie on grid points
    and all edges of~$T_u$ have integer length.
    Clearly, there is no crossing among the edges $uv_1, \dots, uv_k$.
    Suppose for a contradiction that for some $i, j \in \{1,\ldots,k\}$, $i < j$,
    the drawings of $T_{v_i}$ and $T_{v_j}$ intersect.
    By \cref{inv:cone2}, this implies that the cones of $T_{v_i}$ and $T_{v_j}$ intersect.
    Since $i < j$, $uv_i$ has a smaller slope than~$uv_j$.
    For the cone based in~$v_i$ to intersect the cone based in~$v_j$,
    the upper bounding ray of $v_i$'s cone must have a greater slope than the lower bounding ray of $v_j$'s cone.
    However, all Pythagorean triples contained in $\p(v_i)$ correspond to smaller slopes
    then the Pythagorean triples contained in $\p(v_j)$ since they are sorted by angle.
    A similar argument can be made to show that none of the edges $uv_1, \dots, uv_k$ crosses $T_{v_1}, \dots, T_{v_k}$.
    Hence, the combined drawing is planar and \cref{inv:plane2} is satisfied.
    Furthermore, consider the rays emanating from~$u$
    with the slopes corresponding the first and the last Pythagorean triple of~$\p(u)$.
    They contain all edges $uv_1, \dots, uv_k$ and the cones based at $v_1, \dots, v_k$.
    Therefore, \cref{inv:cone2} is satisfied.
    Finally, the maximum Euclidean distances between $u$ and any of $v_1, \dots, v_k$ is~$\frac{\pi^2}{3}\ell$ by \cref{lem:triple_size}.
    Together with \cref{inv:grid2} and the triangle inequality,
    this implies that the Euclidean distance between $u$ and any vertex~$x$ in $T_{v_1}, \dots, T_{v_k}$ is
    \begin{equation*}
        \mathrm{dist}(u, x) < \frac{\pi^2}{3} \ell + \max_{i = 1}^k d(T_{v_i}) \cdot \frac{\pi^2}{3} \ell = \left(\max_{i = 1}^k d(T_{v_i}) + 1\right) \frac{\pi^2}{3} \ell = d(T_u) \frac{\pi^2}{3} \ell.
    \end{equation*}
    This satisfies \cref{inv:grid2}.
    \cref{inv:plane2,inv:cone2,inv:grid2} for the root~$r$ imply the theorem.
    
    \subparagraph*{Running time.}
    Construct $\p_n$ in $\oh(n)$ time using \cref{lem:compute_triples}.
    We next analyze the steps of the algorithm.
    The base case takes a $\oh(1)$ amount of time; otherwise, take $\oh(k)$ time per vertex~$v$ where $k$ is the number of children of~$v$ -- we can assign an arbitrary number of Pythagorean triples to a child
    in $\oh(1)$ time by only specifying the start and end position in an array.
    Because every vertex is the child of at most one other vertex, the overall running time is in~$\oh(n)$.
    Also note that it suffices to save the offsets in the intermediate steps
    and, using them, to compute the final coordinates in the end
    in linear time in a single tree traversal.
\end{proof}

\subsection{General Trees}
\label{sec:trees}

The drawings from \cref{thm:trees-bounded-depth} may use quartic area if we cannot
make any assumptions about the depth or the number of leaves.
Next, we consider a more sophisticated approach
producing a truly integral F\'ary embedding on an integer grid of almost-quadratic size for any tree,
approaching the quadratic lower bound from \cref{thm:stars-lower-bound}.
The main difference from \cref{sec:trees-bounded-depth}
is that the largest subtree gets a short vertical edge. This can give a recursive stacking of vertical edges, which we resolve with \cref{lem:vertical-gap} with a sufficiently large set of Pythagorean triples for each subtree using  the scaled Pythagorean triples in $\q_{n}$
instead of the primitive triples in $\p_{n}$.

\begin{theorem} \label{thm:trees}
    Let $T = (V, E)$ be a (rooted) tree with $n$ vertices.
    There is a truly integral F\'ary embedding on a grid of size $27n \cdot \log n \times 36n \cdot \log n \subseteq \oh(n \log n) \times \oh(n \log n)$,
    which can be found in $\oh(n)$ time.
\end{theorem}

\begin{proof}
If no root is given explicitly,
we can choose an arbitrary vertex as root.
Let this be~$r$.
We start by describing a recursive algorithm for drawing~$T$, which we analyze afterwards.

\begin{figure}[t]
    \centering
    \includegraphics[page=1]{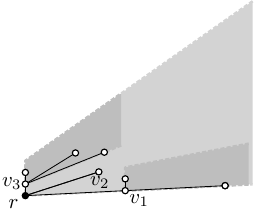}
    \caption{Schematization of our recursive drawing algorithm. The subtree $T_{v_3}$ is the largest among $T_{v_1}, T_{v_2}, T_{v_3}$ and we place $v_3$ at $(0,1)$.}
    \label{fig:tree-adapted}
\end{figure}

\subparagraph{Algorithm.}
For drawing $T$, we only use the first quadrant
and we place the root $r$ at $(0, 0)$.
Let $v_1, \dots, v_k $ be the children of $r$,
determine $\max_{i \in [k]} n_{v_i}$, and find a child for which this maximum is attained;
w.l.o.g., assume that this is $v_k$.
Assign the first child $v_1$ the first $n_{v_1}$ triples from~$\q_n^\circ$, the second child $v_2$ the next $n_{v_2}$ triples and so on.%
\footnote{In order to make \cref{lem:vertical-gap} applicable, we do not assign only $|L(v_i)|$ triples as we did in \cref{thm:trees-bounded-depth}.}
For a vertex~$v$, we let~$\p(v)$ be the ordered set of Pythagorean triples assigned to~$v$.
Place the child $v_k$, which has a largest subtree among all children of~$r$, at $(0,1)$.
For $i \in [k-1]$, place the child~$v_i$ at the coordinates $(x, y)$ where $(x, y, z)$ is the first Pythagorean triple of $\p(v_i)$.
For each child~$v_i$ that is not a leaf of~$T$,
we draw the subtree~$T_{v_i}$ recursively:
in the recursive call, we assume that the coordinate system is centered at the position of~$v_i$ and we use only the Pythagorean triples in~$\p(v_i)$.
For an illustration, see \cref{fig:tree-adapted}.

\subparagraph{Correctness.}
By \cref{lem:triple_size}, choosing $i, j \in \mathbb{N}$ such that $j < i \le \frac{\pi\sqrt{2}}{\sqrt{3}}\sqrt{n}$
is enough to generate the set $\p_{4n}^\eqref{eq:euclid}$ (from which we have constructed $\q_n$) consisting of $2n$ primitive Pythagorean triples of the form $(x_k, y_k, z_k) = (i^2-j^2,2ij,i^2+j^2)$.
For each such primitive Pythagorean triple $P = (x, y, z)$,
we have, for $x$ by the scaling and by \cref{lem:triple_size} and for $y$ by \cref{lem:lower-half-triples},
\begin{align*}
    x &< \frac{4\pi^2}{{3}}{n}, \quad\quad\quad\textrm{ and} \\
    y &= \tan(\alpha(P)) \cdot x < \tan(2 \arctan(1/2)) \cdot  x \le \frac{4x}{3}.
\end{align*}
Consider the step of scaling $P$ (in the generation of~$\q_n$)
with the minimum $c \in \mathbb{N}$
such that $x \cdot c > \frac{4\pi^2}{3}n$.
It is easy to see that we are at most two times larger than the previous bound.
\begin{align}
    x &< 2 \cdot \frac{4\pi^2}{{3}}{n} < 27n, \quad\quad\quad\textrm{ and}
    \label{eq:x-pyth-triple} \\
    y &\le \frac{4x}{3} < \frac{4}{3} \cdot 2 \cdot \frac{4\pi^2}{{3}}{n} < 36n.
    \label{eq:y-pyth-triple}
\end{align}

We next show inductively that for every vertex~$v$ of~$T$, the following invariants hold:
\begin{enumerate}
    \item\label{inv:plane} The drawing of~$T_v$ is a truly integral F\'ary embedding.
    \item\label{inv:trapezoid} The drawing of $T_v$ lies in a region bounded by
    \begin{itemize}
        \item a vertical segment $(p_v,q_v)$ with length $n_v$, where $p_v$ is position of $v$, and
        \item two rays emanating at~$p_v$ resp., $q_v$ with slopes $y/x$ resp., $y' / x'$, where the triples $(x, y, z)$ and $(x', y', z')$ denote the first and the last Pythagorean triple of $\p(v)$, respectively.
    \end{itemize}
    For an illustration of these regions, see also the dark gray trapezoids in \cref{fig:tree-adapted-proof}.
    \item\label{inv:grid} The drawing of $T_v$ has a grid size of at most $w(T_v) \times h(T_v) = 27n \cdot \log(n_v) \times 36n \cdot \log(n_v)$.
\end{enumerate}

\begin{figure}[t]
    \centering
    \includegraphics[page=2]{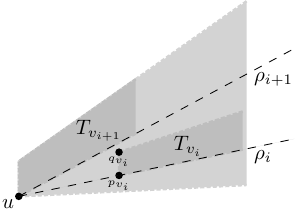}
    \caption{While every subtree itself lies in a trapezoid (\cref{inv:trapezoid}), viewed from the parent $u$, they are separated by cones.}
    \label{fig:tree-adapted-proof}
\end{figure}

\subparagraph{Base case: a leaf.} We start with the base case that $v$ is a leaf.
The drawing of~$T_v$ contains only $v$ being placed at~$(0, 0)$.
Trivially, 
all invariants hold.

\subparagraph{Recursion: combining child components of a vertex $u$.}

We consider the combination step where $u$ is the root
of~$T_u$ and $u$ has children $v_1, \dots, v_k$ for some $k \in [n-1]$.
We assume that $T_{v_1}, \dots, T_{v_k}$
are drawn according to the three invariants.
We place $u$ at $(0, 0)$ and, for each $i \in [k-1]$,
we place the drawing of~$T_{v_i}$ such that $v_i$ lies on $(x, y)$
where $(x, y, z)$ is the first Pythagorean triple of $\p(v_i)$.
Furthermore, we place the drawing of~$T_{v_k}$ such that $v_k$ lies on $(0,1)$.
Here, we only have applied a translation of each~$T_{v_i}$ from a grid point to another grid point;
see also the child components of $r$ in \cref{fig:tree-adapted}.
All vertices of~$T_{v_i}$ as well as $u$ lie on grid points
and all edges of~$T_u$ have integer length
(the edges incident to~$u$ since we have used Pythagorean triples, the other edges by induction).

For planarity, we argue that the regions of the subtrees defined by \cref{inv:trapezoid} do not intersect.
We show that each subtree $T_{v_i}$ lies inside a cone bounded by the rays $\rho_i, \rho_{i+1}$ rooted at $u$,
where $\rho_i, \rho_{i+1}$ have the slopes of the first Pythagorean triples~$P_{v_i}$ and~$P_{v_{i+1}}$ of $\p(v_i)$ and $\p(v_{i+1})$; 
see also \cref{fig:tree-adapted-proof}.
For this, it is enough to show that the endpoint~$q_{v_i}$ of the vertical segment described in \cref{inv:trapezoid} lies strictly below $\rho_{i+1}$.
Recall that $P_{v_{i+1}}$ is the $n_{v_i}$-th Pythagorean triple after $P_{v_i}$.
Therefore, when we plug this into \cref{lem:vertical-gap},
we know that the vertical distance between $p_{v_i}$ and $\rho_{i+1}$ is greater than~$n_{v_i}$.
Hence, $q_{v_i}$, which is $n_{v_i}$ units above $p_{v_i}$ lies strictly below~$\rho_{i+1}$.
Together with the fact discussed above that all edges have integer edge lengths and lie on integer coordinates, \cref{inv:plane} is implied.

For \cref{inv:trapezoid}, only consider subtrees $T_{v_1}$ and $T_{v_k}$ of $u$, as we already shown that,
for $i \in [k-1] \setminus \{1\}$, the subtree $T_{v_i}$ lies in a cone  bounded by the regions of~$T_{v_{i-1}}$ and~$T_{v_{i+1}}$.
Clearly, the edge $uv_1$ and the region of $T_{v_1}$ from a ray emanating in the position of~$u$ (i.e., $p_u$)
with the slope of the first Pythagorean triple of~$\p(u)$.
Now consider the upper boundary.
Since $v_k$ has coordinates $(0, 1)$ and, by induction, $q_{v_k}$ lies $n_{v_k}$ units above $v_k$,
the y-coordinate of $q_{v_k}$ is $n_{v_k} + 1 \le n_u$, which is the y-coordinate of $q_u$.
Moreover, the ray bounding the region of $T_{v_k}$ from above lies below or on the ray emanating in~$q_{u}$
with the slope of the last Pythagorean triple of~$\p(u)$. 

Finally, we show an upper bound on the grid size
by showing that \cref{inv:grid} holds.
We will use the following fact for $x, y, c \ge 1$ and $x \ge y$:
\begin{align}
    \label{eq:logs-upper-bound}
    \left(\frac{x}{x-1}\right)^x \ge e = 2.718\ldots \ge 2^{\frac1c} &\Rightarrow \log \left( \left(\frac{x}{x-1}\right)^x \right) \ge \frac1c \Rightarrow \nonumber x \left(\log(y)-\log(y-1)\right) \ge \frac1c \\
    &\Rightarrow c x \log(y - 1) + 1 \le c x \log(y)
\end{align}
With respect to the drawing of the subtree $T_{v_k}$ (which has at most $n_u - 1$ vertices),
the drawing of~$T_{u}$ has at least the same width and the height increases by exactly one
if the highest vertex of the drawing of~$T_u$ lies in~$T_{v_k}$.
With respect to the drawing of the other subtrees,
the resulting drawing may have their width and height plus the width and height of a Pythagorean triple from~$\q_n$
(as analyzed in Inequalities~\eqref{eq:x-pyth-triple} and~\eqref{eq:y-pyth-triple}),
but note that each of these subtrees has at most $\frac{n_u - 1}{2}$ vertices.
Hence, by \cref{inv:grid} for the subtrees, we have
\begin{align*}
    w(T_u) &\le \max\left(27 n \cdot \log (n_u - 1), 27 n \cdot \log \left(\frac{n_u - 1}{2}\right) + 27\right) \le 27 n \cdot \log (n_u) \textrm{ and} \\
    h(T_u) &\le \max\left(36 n \cdot \log (n_u - 1) + 1, 36 n \cdot \log \left(\frac{n_u - 1}{2}\right) + 36\right)
    \overset{\eqref{eq:logs-upper-bound}}{\le} 36 n \cdot \log (n_u).
\end{align*}
Thus, \cref{inv:grid} holds for $T_u$. 
That all invariants hold from the root $r$, implies the theorem.

\subparagraph*{Running time.}
Construct $\p_{4n}$ and from there $\q_n^\circ$ in $\oh(n)$ time 
using \cref{lem:compute_triples}.
We next analyze the steps of the algorithm.
The base case takes a $\oh(1)$ amount of time; otherwise, take $\oh(k)$ time per vertex~$v$ where $k$ is the number of children of~$v$ -- we can assign an arbitrary number of Pythagorean triples to a child
in $\oh(1)$ time by only specifying the start and end position in an array.
Because every vertex is the child of at most one other vertex, the overall running time is in~$\oh(n)$.
\end{proof}

\section{Truly Integral F\'ary Embeddings of Cacti in Polynomial Area}
\label{sec:cacti}

We now consider a slightly more complex graph class, cacti.
While cacti can be seen as a first step towards extending our result on trees to outerplanar graphs,
they still admit a simple tree-like decomposition, which allows us to extend the ideas from \cref{sec:all-trees},
with the primary complication arising from triangles.
To handle many triangles incident to the same vertex,
we currently only see an approach requiring $\Theta(n^2) \times \Theta(n^2)$ area.
In a tree-like traversal of a cactus with depth~$\Theta(n)$,
this gives an area of $\Theta(n^3) \times \Theta(n^3)$.
Using a similar trick as in \cref{sec:trees},
we  reduce this to $\oh(n^2 \log n) \times \oh(n^2 \log n)$:
We draw in each recursive step a triangle of a largest subcactus
with the Pythagorean triple $(3, 4, 5)$, which can be seen as the cyclic
counterpart of the short vertical edge $(0, 1)$.
We must be careful  not to stack $(3, 4, 5)$-triangles
next to each other as this would cause overlaps.
To this end, we introduce a vertex coloring where
we color a vertex red if it must not have such a triangle as a child.

\subparagraph*{Definitions.}
Let $G$ be a cactus with root~$r$.
The structure of~$G$ can be represented using a \emph{block-cut tree}
whose nodes are the blocks (biconnected components) and cut vertices (vertices that separate the graph when removed) of~$G$
and whose edges are the containments of cut vertices in blocks.
The block-cut tree of a cactus can be computed in linear time~\cite{DBLP:journals/cacm/HopcroftT73}.
This tree structure allows us in the following algorithm a tree traversal starting in the root~$r$.
When removing a vertex~$v$, $G$ would decompose into a set
of connected components $\mathcal{G}_v$.
We denote by $G_v$ the subcactus of $G$ rooted at $v$, i.e.,
the graph that remains after removing from~$G$
the connected component of~$\mathcal{G}_v$ that contains~$r$
(if $v$ is a cut vertex, then this is the subtree in the block-cut tree rooted at~$v$).
For a cycle $C$ in $G$, we let the \emph{origin} of~$C$
be the unique vertex $s_C$ on $C$
that any vertex on~$C$ includes on any path to the root~$r$;
or if $C$ contains~$r$, then $s_C = r$.
Similarly, for an edge~$e$ that does not lie on any cycle,
we let the \emph{origin} of~$e$
be the vertex~$s_e$ of~$e$ that is closer to~$r$;
we call the other vertex of~$e$ the \emph{target} of~$e$.
A block $B$ is either a cycle or an edge that does not lie on any cycle.
We say that a block $B$ is a \emph{child} of a vertex $v$ if $v$ is the origin of~$B$.
The subcactus $G_B$ of a block~$B$ is the union of~$B$ and all subcacti
of all vertices of $B$ except for the subcactus~$G_{s_B}$ of
the origin~$s_B$.
Similar to subtrees,
denote by $C(B)$, $L(B)$, and~$n_B$ (or $C(v)$, $L(v)$, and $n_v$) the number of cycles,
leaves, and vertices in the subcactus of a block~$B$ (or a vertex~$v$), respectively.

\begin{theorem} \label{thm:cactus}
    Let $G = (V,E)$ be a (rooted) cactus with $n$ vertices.
    There is a truly integral F\'ary embedding on a grid of size $15052n^2 \log n \times 20022 n^2 \log n \subseteq \oh(n^2 \log n) \times \oh(n^2 \log n)$,
    which can be found in $\oh(n)$ time.
\end{theorem}

\begin{proof}
Let $r$ be the root.
If no root is given explicitly, we can choose an arbitrary vertex.

\subparagraph{Algorithm.}
We place $r$ at $(0,0)$ and we use at most $\ell+2c$ primitive Pythagorean triples in $\q_{\ell+2c, 4}$ for drawing $G$,
where $\ell$ is the number of leaves and $c$ is the number of cycles in~$G$.
As in the algorithm for arbitrary trees, we restrict ourselves to the first quadrant.

Let~$\mathcal{B}_r$ be the set of blocks that are the children of~$r$.
We will do some re-ordering and flipping of the blocks
to bring largest subcacti to special positions.
Among~$\mathcal{B}_r$, find the block with the largest subcactus
and denote it by $B_k$ where $k = |\mathcal{B}_r|$;
the other blocks $B_1, \dots B_{k-1}$ in~$\mathcal{B}_r$
can be kept in input order.
Moreover, for $i \in [k]$,
if $B_i$ is a cycle, fix an order of its vertices along the cycle
starting in~$r$
such that the last $\lfloor n / 2 \rfloor$ vertices of~$B_i$
contain a largest subcactus among all subcacti
whose origin are the vertices of~$B_i$ (ignoring the origin of~$B_i$).

Assign $B_1$ the first $2C(B_1) + L(B_1)$ triples from~$\q_{\ell+2c, 4}^\circ$,
assign $B_2$ the next $2C(B_2) + L(B_2)$ triples from~$\q_{\ell+2c, 4}^\circ$ and so on.
We will later also assign Pythagorean triples to vertices.
For a block~$B$ (or a vertex~$v$), we let~$\p(B)$ (or $\p(v)$)
be the ordered set of Pythagorean triples assigned to~$B$ (or~$v$, respectively).
Let $i\in [k]$ and check the type of~$B_i$.
For an overview of our approach, see \cref{fig:cactus}.

First assume that $B_i$ is an edge.
We proceed as in the tree algorithm from \cref{sec:trees-bounded-depth},%
\footnote{The algorithm from \cref{sec:trees}  would not asymptotically improve area due to the way we draw cycles.}
that is, place the target~$t$ of the edge at $(x, y)$
where $(x, y, z)$ is the first Pythagorean triple of $\p(B_i)$.
Then, let $\p(t) = \p(B_i)$ and recurse on~$t$.
Now assume that $B_i$ is a cycle of length~$d$.
We use a special handling for triangles
that belong to the largest subtree~$B_k$
(drawing it with side lengths $(3,4,5)$),
but we must not apply this special handling to the smaller child
of~$B_k$ again because this would result in overlaps.
Hence, we will color the vertices that cannot
use the triangle $(3, 4, 5)$ for any of its children red.
Initially, we set all vertices to be white.

\begin{figure}[t]
    \centering
    \begin{subfigure}{0.26\textwidth}
        \centering
        \includegraphics[page=1]{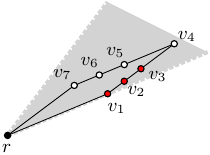}
        \subcaption{}
        \label{fig:cactus-cycle:1}
    \end{subfigure}
    \hfill
    \begin{subfigure}{0.26\textwidth}
        \centering
        \includegraphics[page=2]{cactusNEW.pdf}
        \subcaption{}
        \label{fig:cactus-cycle:2}
    \end{subfigure}
    \hfill
    \begin{subfigure}{0.26\textwidth}
        \centering
        \includegraphics[page=3]{cactusNEW.pdf}
        \subcaption{}
        \label{fig:cactus-cycle:3}
    \end{subfigure}
    \hfill
    \begin{subfigure}{0.15\textwidth}
        \centering
        \includegraphics[page=4]{cactusNEW.pdf}
        \subcaption{}
        \label{fig:cactus-cycle:4}
    \end{subfigure}
    \caption{Drawings of an (a) even- or (b) odd-length cycle of length greater than~3.
    Drawing of a triangle as (c) a large triangle or (d) a triangle with side lengths (3, 4, 5).}
    \label{fig:cactus-cycle}
\end{figure}

\subparagraph*{Case $d \ge 4$:}
We draw~$B_i$ as a parallelogram;
see \cref{fig:cactus-cycle:1,fig:cactus-cycle:2}.
Let $\langle r = s_{B_i} = v_0, v_1, \dots, v_{d-1} \rangle$ be the vertices
of the cycle in order,
and let $m = \lfloor d / 2 \rfloor$ be the middle index.
The corners of our parallelogram will be $v_0$, $v_1$, $v_m$ and $v_{d-1}$,
and the sides will use the slopes $y/x$ and $y'/x'$
where $(x, y, z)$ is first Pythagorean triple of~$\p(B_i)$
and $(x', y', z')$ is a middle Pythagorean triple of~$\p(B_i)$ defined below.
More precisely, starting with the second Pythagorean triple,
we assign $2C(v_1) + L(v_1)$ Pythagorean triples of~$\p(B_i)$ to~$v_1$,
the next $2C(v_2) + L(v_2)$ Pythagorean triples of~$\p(B_i)$ to~$v_2$, etc.\
until we have assigned Pythagorean triples to~$v_{m-1}$.
The next Pythagorean triple is~$(x', y', z')$.
Then, we assign the next $2C(v_m) + L(v_m)$ Pythagorean triples to~$v_m$,
the next $2C(v_{m+1}) + L(v_{m+1})$ Pythagorean triples to~$v_{m+1}$, etc.\
until we have assigned Pythagorean triples to~$v_{d-1}$.
For $j \in [m-1]$, draw $v_j$ at $(m x + (j-1) x', m y + (j-1) y')$.
Draw $v_m$ at $(m (x + x'), m (y + y'))$.
For $j \in \{m+1, \dots, d-1\}$, draw $v_j$ at $(m x' + (d-1-j) x, m y' + (d-1-j) y)$.
The vertices $v_1, \dots, v_{m-1}$ lie on the bottom side of the parallelogram
where we cannot directly attach the triangle $(3,4,5)$ without entering or intersecting the parallelogram.
Hence, we color the vertices $v_1, \dots, v_{m-1}$ red.
For $j \in [d-1]$, we then go into recursion on~$v_j$.

\subparagraph*{Case $d = 3$ and ($i < k$ or $v_i$ is red):}
We draw $B_i = \langle r = s_{B_i} = v_0, v_1, v_2 \rangle$ as a triangle
whose bottommost point is~$v_0$ and whose top edge $v_1v_2$ is horizontal;
see \cref{fig:cactus-cycle:3}.
We use the first Pythagorean triple~$(x, y, z)$ of~$\p(B_i)$ to draw the triangle,
we assign the next $2C(v_1) + L(v_1)$ Pythagorean triples to~$v_1$,
we use the next Pythagorean triple~$(x', y', z')$ also for drawing the triangle,
and assign the remaining $2C(v_2) + L(v_2)$ Pythagorean triples to~$v_2$.
We place $v_1$ at $(x y', y y')$ and $v_2$ at $(x' y, y y')$.
Finally, we go into recursion on~$v_1$ and~$v_2$.

\subparagraph*{Case $d = 3$ and $i = k$ and $v_i$ is white:}
Let $\langle r = s_{B_k} = v_0, v_1, v_2 \rangle$
and draw the triangle as shown in \cref{fig:cactus-cycle:4}.
More precisely, place $v_1$ at $(3, 4)$,
assign $v_1$ the first $2C(v_1) + L(v_1)$ Pythagorean triples in~$\p(B_k)$,
place $v_2$ at $(0, 4)$, and
assign $v_2$ the next $2C(v_2) + L(v_2)$ Pythagorean triples in~$\p(B_k)$.
We go into recursion on~$v_1$ and~$v_2$,
but we must make sure that we do not use the triangle $(3,4,5)$
for a child of~$v_1$.
Hence, we color~$v_1$ red before going into recursion.

\subparagraph{Recursion.}
Let $u$ be a vertex at which we go into recursion.
We remark that if $u$ is not a cut vertex, then the recursion stops there
and we have not assigned any Pythagorean triple to~$u$,
or $u$ is a leaf.
As in the tree algorithms,
we translate the drawing in each recursive call on a vertex~$u$
such that $u$ is placed in~$(0, 0)$
and we only use the Pythagorean triples in~$\p(u)$.

\begin{figure}[t]
    \centering
    \includegraphics[page=6]{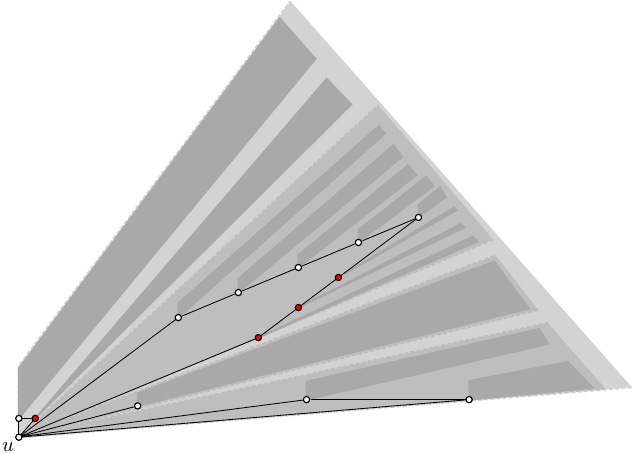}
    \caption{Schematization of our recursive algorithm in \cref{sec:cacti}:
        we draw a subcactus rooted at~$u$
        with $u$ being placed at $(0,0)$.
        Here, we have four children blocks: an edge, two triangles, and a cycle of length~9.
        The dark shaded areas contain the drawing of subcacti as inductively guaranteed by \cref{inv:cac:trapezoid};
        $u$ is a white vertex because one of its children is drawn as the triangle (3, 4, 5).}
    \label{fig:cactus}
\end{figure}

\subparagraph{Correctness.}
We next show inductively that for every vertex~$u$ of~$T$, the following invariants hold
making it a planar truly integral F\'ary embedding.
For completeness, we assume that the next Pythagorean triple after the last Pythagorean triple of~$\q_{\ell+2c, 4}^\circ$
is~$(3,4,5)$,
which we can do because $\arctan (4/3) > \alpha(P)$
for every $P \in \q_{\ell+2c, 4}^\circ$ by \cref{lem:lower-half-triples}.
\begin{enumerate}
    \item\label{inv:cac:plane} The drawing of~$G_u$ is a truly integral F\'ary embedding.
    \item\label{inv:cac:trapezoid}
    If $u$ is a white vertex,
    then the drawing of $G_u$ lies in a region bounded by
    \begin{itemize}
        \item a vertical segment $(p_u,q_u)$ with length $4 C(u)$, where $p_u$ is the position of $u$, and
        \item two rays emanating at~$p_u$ resp., $q_u$ with slopes $y/x$ resp., $y' / x'$, where the triples $(x, y, z)$ and $(x', y', z')$ denote the first and the last Pythagorean triple of $\p(u)$, respectively.
    \end{itemize}
    If $u$ is a red vertex,
    then the drawing of $G_u$ lies in a cone bounded by
    \begin{itemize}
        \item a bottom ray emanating at the position of~$u$ and having slope $y/x$,
        where the triple $(x, y, z)$ denotes the first Pythagorean triple of $\p(u)$, and
        \item a top ray emanating at the position of~$u$ and having slope $y'/x'$,
        where the triple $(x', y', z')$ denotes the next Pythagorean triple in $\q_{\ell+2c, 4}^\circ$ after the last Pythagorean triple of $\p(u)$.
        The drawing of $G_u$ does not touch the top ray.
    \end{itemize}
    For an illustration of these regions, see the dark gray trapezoids and cones in \cref{fig:cactus}.
\end{enumerate}

\subparagraph{Base case: $G_u$ consists of only~$u$.} We start with the base case that $u$ is a leaf or $u$ is a vertex of a cycle where nothing is attached.
In both cases, $G_u$ consists of only~$u$, which is placed at~$(0, 0)$.
Trivially, \cref{inv:cac:plane,inv:cac:trapezoid} hold.

\subparagraph{Recursion: combining child components of a vertex $u$.}

We consider the combination step where $u$ is the root
of~$G_u$ and $u$ has children $B_1, \dots, B_k$ for some $k \in [n-1]$.
In turn, each child $B_i$ ($i \in [k]$)
consists of a sequence of vertices $\langle u = s_{B_i} = v_0, v_1, \dots, v_{|B_i| - 1} \rangle$ forming a cycle or an edge if $|B_i| = 2$.
We assume that $G_{v_1}, \dots, G_{v_{|B_i| - 1}}$
are drawn according to \cref{inv:cac:plane,inv:cac:trapezoid}.
First observe that the vertices in $B_i$ are placed on
integer coordinates because they are integer multiples
of the entries in the Pythagorean triples.
The following claim helps us to argue that the children blocks do not intersect.
We prove it in the appendix using mainly \cref{lem:vertical-gap}
as we did similarly in the proof of \cref{thm:trees}.

\begin{claim}
\label{clm:children-in-cone}
    For $i \in [k - \delta]$ with $\delta = 0$ if $u$ is red and $\delta = 1$ if $u$ is white,
    $G_{B_i}$ is inside the cone bounded by rays~$\rho_i$ and $\rho_{i+1}$ emanating in $(0,0)$
    so that no vertex or edge of~$B_i$ touches $\rho_{i+1}$.
    Rays $\rho_i$ and $\rho_{i+1}$ have slopes $y/x$ and $y'/x'$, respectively,
    where $P = (x, y, z)$ is the first Pythagorean triple of $\p(B_i)$ and
    $P' = (x', y', z')$ is the next Pythagorean triple in $\q_{\ell+2c, 4}^\circ$
    after the last Pythagorean triple of~$\p(B_i)$
    (i.e., for $i < m$, the first triple in~$\p(B_{i+1})$).
\end{claim}

\begin{claimproof}
    First consider the case that $B_i$ is an edge~$uv$.
    We have reserved a sequence of $2C(B_i) + L(B_i)$ Pythagorean triples from~$\q_{\ell+2c, 4}^\circ$ for~$G_{B_i}$.
    Hence, we know by \cref{lem:vertical-gap} that the vertical distance~$\Delta$ between~$p_v$ and~$\rho_{i+1}$
    is $\Delta > 4 (2C(B_i) + L(B_i))$.
    Since the edge~$uv$ lies on $\rho_i$ and \cref{inv:cac:trapezoid} guarantees us that
    $G_v$ lies in a region that is bounded on the left by a vertical segment
    of height at most~$4 C(v) = 4 C(B_i) < 4 (2C(B_i) + L(B_i)) < \Delta$,
    it follows that the drawing of $G_{B_i}$ lies inside of this cone.
    
    Now consider the case that $B_i$ is a triangle~$\langle u = v_0, v_1, v_2 \rangle$.
    We know that $v_1$ is placed at $(x \hat{y}, y \hat{y})$ and $v_2$ is placed at $(\hat{x} y, y \hat{y})$
    where $(\hat{x}, \hat{y}, \hat{z})$ is the $(2C(v_1) + L(v_1))$-th Pythagorean triple after $(x, y, z)$.
    Note that by again using \cref{lem:vertical-gap},
    we know that the vertical distance between $v_1$ and the ray~$\hat{\rho}$ emanating in $(0, 0)$ and having slope~$\hat{y}/\hat{x}$
    is greater than $4(2C(v_1) + L(v_1))$.
    Moreover, we know by \cref{inv:cac:trapezoid} that the left vertical part of the region in which $G_{v_1}$ is drawn
    has a height of $4 C(v_1)$, which is smaller, and, hence, this region lies strictly below $\hat{\rho}$.
    By \cref{inv:cac:trapezoid} we also know that $G_{v_2}$ is drawn above the ray $\hat{\rho}$.
    Similarly, we apply \cref{lem:vertical-gap} to show that $G_{v_2}$
    lies strictly below the ray~$\rho_{i+1}$ since we have assigned $2C(v_2) + L(v_2)$ Pythagorean triples to~$v_2$.
    Hence, the drawings of $G_{v_1}$ and $G_{v_2}$ do not intersect each other and lie between $\rho_i$ and $\rho_{i+1}$.
    As the triangle $\langle v_0, v_1, v_2 \rangle$ extends below the drawings of $G_{v_1}$ and $G_{v_2}$,
    the claim is satisfied for $G_{B_i}$.
    
    Finally, consider the case that $B_i$ is a cycle $\langle u = v_0, v_1, \dots, v_{d-1} \rangle$ of length~$d > 3$.
    First observe that $B_i$ is drawn as parallelogram with counter-clockwise corners $v_0, v_1, v_m, v_{d-1}$
    (with $m = \lfloor d / 2 \rfloor$)
    where the line segments $\overline{v_0v_1}$ and $\overline{v_1v_m}$ constitute the bottom side
    and the line segments $\overline{v_0v_{d-1}}$ and $\overline{v_{d-1}v_m}$ constitute the top side of the parallelogram.
    For $j \in [m-1]$, each vertex~$v_j$ is red and lies on the bottom side of the parallelogram.
    By \cref{inv:cac:trapezoid}, $v_j$ lies in a cone whose upper boundary has a slope
    being smaller than the slope of $\overline{v_1v_m}$, which comes from the Pythagorean triple of~$\p(B_i)$
    with the largest angle.
    Hence, the drawing of $G_{v_j}$ does not intersect the parallelogram.
    Moreover, none of the drawings of these subcacti intersect each other since their bounding rays are ordered by slope
    forming a fan-like structure as each consecutive one starts above the other.
    For each $j \in \{m, \dots, d-1\}$, each vertex~$v_j$ is white and lies on the top side of the parallelogram.
    By \cref{inv:cac:trapezoid} it is drawn in a trapezoidal region
    whose left vertical segment has height~$4C(v_j)$.
    This region extends upwards and has a greater slope than the line segment $\overline{v_{d-1}v_m}$;
    therefore, the drawing of $G_{v_j}$ does not intersect the parallelogram.
    Moreover, note that by \cref{lem:vertical-gap}, the regions of these subcacti are pairwise disjoint.
    More precisely, consider the regions of $G_{v_j}$ and $G_{v_{j+1}}$ for some $j \in \{m, \dots, d-1\}$.
    The distance between the left foot point $p_{v_{j+1}}$ and the right foot point $p_{v_j}$
    is determined by the Pythagorean triple $(x, y, z)$:
    it is at least x-distance~$x$ and at least y-distance~$y$.
    In $\q_{\ell+2c, 4}^\circ$, there are at least $|\p(v_j)|$ Pythagorean triples between $(x, y, z)$ and the first Pythagorean
    triple in $\p(v_{j+1})$, which gives the slope of the ray~$\rho_{j+1}$ bounding the region of $G_{v_{j+1}}$ from below.
    Applying \cref{lem:vertical-gap} gives a vertical distance greater than $4(2C(v_j) + L(v_j))$ between~$p_{v_j}$ and~$\rho_{j+1}$
    Observe that, if no two consecutive regions intersect, also the other regions do not intersect.
    Moreover, the line with the slope of the first Pythagorean triple of $\p(v_m)$ going through $v_m$
    separates all drawing regions of $G_{v_1}, \dots, G_{v_{m-1}}$ from those of $G_{v_m}, \dots, G_{v_{d-1}}$.
    
    We have shown that the drawing of~$G_{B_i}$ is crossing-free.
    It remains to show that it lies inside of the cone bounded by $\rho_i$ and $\rho_{i+1}$.
    First note that the line segment $\overline{v_0v_1}$ lies on~$\rho_i$ and
    the line segment $\overline{v_1v_m}$ and the region of the drawing of $G_{v_1}$, which both attach at $p_{v_1}$,
    extend above~$\rho_i$.
    As no edge of~$G_{B_i}$ is drawn with a slope smaller than that of~$\rho_i$,
    the drawing of~$G_{B_i}$ lies on and above~$\rho_i$.
    For the upper boundary~$\rho_{i+1}$, we argue again with \cref{lem:vertical-gap}.
    The height of the left vertical segment bounding the region of~$G_{v_j}$ for any $j \in [d-1]$
    is less than the vertical gap guaranteed by the lemma since
    the Pythagorean triple defining the slope~$\rho_{i+1}$ is the $(\sum_{j=m}^{d-1} (2C(v_j) + L(v_j)))$-th one
    after the Pythagorean triple from which $\overline{v_0v_1}$ is constructed.
\end{claimproof}

\cref{clm:children-in-cone} guarantees  that all but potentially the last child of~$u$
lie above the ray with slope $y/x$ where $(x, y, z)$ is the first Pythagorean triple of~$\p(u)$.
Further, the cones for the children (except for potentially the last child)
do not overlap, i.e., their drawing is planar.

If $u$ is a red vertex, then \cref{clm:children-in-cone} also holds for the last child~$B_k$.
Then, \cref{inv:cac:plane} is fulfilled and
everything lies also strictly under the ray $(x', y', z')$
where the triple $(x', y', z')$ denotes the next Pythagorean triple in $\q_{\ell+2c, 4}^\circ$ after the last Pythagorean triple of $\p(u)$.

If $u$ is a white vertex, we distinguish the type of child~$B_k$.
If $B_k = uv_1$ is an edge, then $p_{v_k} = (\hat{x}, \hat{y})$
where $(\hat{x}, \hat{y}, \hat{z})$ is the first Pythagorean triple of~$\p(B_k)$.
Neither the drawing of $G_{v_1}$ (by \cref{inv:cac:trapezoid})
nor the edge $(uv_1)$ enter the cones of the other blocks as defined by \cref{clm:children-in-cone}.
Clearly, everything lies in the region going from $(0, 0)$ vertically up to $(0, 4 C(v))$
and using a ray with the slope of the last Pythagorean triple of~$\p(B_k)$ then.
If $B_k = \langle u, v_1, v_2, \dots \rangle$ is a cycle of length greater than~3,
a similar argument holds since the edge $uv_{|B_k|-1}$ uses a Pythagorean triple from~$\p(B_k)$
and so does the upper boundary of the region of the drawing of~$G_{v_j}$ for each $j \in [|B_k| - 1]$.
If $B_k = \langle u, v_1, v_2 \rangle$ is a triangle,
then $\langle u, v_1, v_2\rangle$ is drawn with side lengths $(3,4,5)$
such that $p_{v_1} = (3, 4)$ and $p_{v_2} = (0, 4)$.
Since in this case $v_1$ is a red vertex, we know by \cref{inv:cac:trapezoid}
that $G_{v_1}$ is drawn in a cone whose lower boundary has a slope greater than
the slope of the upper boundaries of the cones bounding the other blocks as defined by \cref{clm:children-in-cone}.
Also the upper boundary of the cone containing the drawing of~$G_{v_1}$
has a smaller slope than the region containing~$G_{v_2}$.
By \cref{inv:cac:trapezoid}, the drawing of~$G_{v_2}$ lies in a region
bounded on the left by a vertical segment of height $4C(v_2)$.
\cref{inv:cac:trapezoid} allows a height of $4 C(u)$ for the vertical segment on the left
bounding the region where $G_u$ is drawn.
This suffices to stack the vertical edge $uv_2$ of height~4
and the drawing of $G_{v_2}$ since $C(u) \ge C(v_2) + 1$.
According to~\cref{inv:cac:trapezoid}, the upper ray of the region containing~$G_u$
has at least the slope of the upper ray of the region containing~$G_{v_2}$.
We conclude that both \cref{inv:cac:plane,inv:cac:trapezoid} are fulfilled for~$G_u$
completing the inductive step.

\subparagraph{Grid size.}
For an upper bound of the grid size, note that similar to \cref{sec:trees}
(Inequalities~\eqref{eq:x-pyth-triple} and~\eqref{eq:y-pyth-triple}), for every triple $P=(x,y,z) \in \q_{\ell+2c, 4}^\circ$ we have $x \le 4 \cdot \frac{8\pi^2}{3}n < 106n$ and $y \le 4 \cdot \frac{4}{3} \cdot \frac{8\pi^2}{3}n < 141n$ since $\ell+2c \le n$.
Next, we analyze the different cases that can occur in the recursion of the drawing algorithm.
If a block is a single edge, it contributes exactly the height and width of one triple in $\q_{\ell+2c, 4}^\circ$ to the total size.
If a block is a triangle drawn with side lengths $(3,4,5)$, its contribution is trivial as well.
The other triangles may need much more space, namely
they have a width of up to $106\cdot 141 n^2$ and a height of $141^2 n^2$ after scaling the first Pythagorean triple used by the height of the second Pythagorean triple.
We call them \emph{large} triangles.
The larger cycles of length $d \ge 4$ are drawn as parallelograms.
They are constructed so that the $\lfloor d/2 \rfloor$-th vertex of that cycle is assigned to the maximal x- and y-coordinate by scaling two different Pythagorean triples by $\lfloor d/2 \rfloor$ and append them. Thus, in this case the used grid size is at most $106nd \times 141nd$. Taking into account that one of the vertices of that cycle is its origin itself, we get in the recursion a size of $(106nd \times 141nd)/(d-1) \le (4/3) (106n \times 141n)$ per additional vertex.

To put everything together, we make use of the property that in our construction,
large triangles belong to subcacti with relatively few vertices.
More in detail, consider for the width the vertex $v$ with the maximal x-coordinate among all vertices.
Traversing a path $\Pi$ from $r$ to $v$, we count the number of large triangles.
Consider a block $B$ encountered when traversing $\Pi$.
There are three reasons why a child $C$ of a vertex $u$ of $B$ might be a large triangle.
\begin{enumerate*}[label=\textcolor{lipicsGray}{\textsf{\textbf{\arabic*.}}}]
    \item there is another child $C'$ of $u$ with subcactus $G_{C'}$ containing more vertices than~$G_C$,
    \item $B$ is a $(3,4,5)$-triangle $\langle s_B, u, w\rangle$ with origin $s_B$ and the subcactus $G_w$ has at least
    as many vertices as~$G_u$, or
    \item $B$ is a cycle of length greater than~3 drawn as a parallelogram,
    $u$ lies on the bottom side of the parallelogram,
    and there is a vertex $w$ on the top side of the parallelogram
    such that the subcactus $G_w$ has at least as many vertices as~$G_u$.
\end{enumerate*}
In all three cases $n_C \le n_B / 2$ holds.
Thus, there are at most $\log n$ large triangles along~$\Pi$.

We add up all contributions to the width $w$ from the at most $n$ different blocks along~$\Pi$:
\[w \le \log n \cdot 106 \cdot 141 n^2 + (n-\log n) \cdot \frac{4}{3}\cdot 106n \le 106n^2 \cdot(141 \log n + \frac{4}{3}) \overset{(n \ge 3)}\le 106 \cdot 142 n^2\log n \]
Analogously, we have for the height $h$:
\[h \le \log n \cdot 141^2 n^2 + (n-\log n) \cdot \frac{4}{3}\cdot 141n \le 141n^2 \cdot(141 \log n + \frac{4}{3}) \overset{(n \ge 3)}\le 141 \cdot 142 n^2\log n \]
For $n < 3$, the grid bound trivially holds.

\subparagraph*{Running time.}
Similar to the previous algorithms,
most steps of the algorithm straight-forwardly admit a linear-time implementation
including the generation of the primitive Pythagorean triples (see \cref{lem:compute_triples}).
This might not be completely clear for the re-ordering and flipping of the children blocks.
To achieve linear time here,
we first traverse the block-cut tree in linear time in a post-order traversal to
save at each block~$B$ and at each vertex~$v$ the numbers of cycles, leaves, and vertices in their subcacti.
Then, we traverse again each vertex once, find a largest child block and put it to the end of the list of children.
Also, we traverse each (cycle) block once, find a vertex with a largest subcactus and flip the list of vertices if needed.
\end{proof}

We did not work on further optimizing the constants in \cref{thm:cactus}.
Also note that only triangular blocks contribute possibly $\omega(n) \times \omega(n)$ to the total grid size.
Thus, if there are no triangular blocks, our analysis immediately yields:

\begin{remark}\label{rmk:triangle-free}
    If there are no triangles,
    the grid size reduces  to $\oh(n^2) \times \oh(n^2)$.
\end{remark}

\section{Open Problems}

\begin{enumerate}
    \item Is Harborth's conjecture true for all planar graphs?
    \item Can we improve the algorithm for cacti to avoid excessively large triangles?
    \item Compute lower/upper bounds for truly integral F\'ary embeddings of larger graph classes.
    \item Are Harborth's and Kleber's conjectures equivalent? What if we require polynomial area for truly integral F\'ary embeddings?
\end{enumerate}

\bibliography{int-fary-bib}

\end{document}